\documentclass[journal]{IEEEtran}
\usepackage{cite}
\usepackage{graphicx}
\usepackage{amsmath,amsfonts}
\usepackage{algorithm}
\usepackage{algorithmic}
\usepackage{array}
\usepackage{amsthm}
\usepackage{amssymb}
\usepackage[caption=false,font=footnotesize]{subfig}
\usepackage{textcomp}
\usepackage{stfloats}
\usepackage{enumerate}
\usepackage{verbatim}
\usepackage{balance}
\newtheorem{proposition}{Proposition}
\usepackage{hyperref}
\usepackage{orcidlink}
\begin{document}

\title{V-RIS: Virtual-Aperture DoA Estimation with Sparse RIS}


\author{Fenghao Zheng\,\orcidlink{0009-0007-2213-8296},
Fuhai Wang\,\orcidlink{0000-0001-9044-6855},
Zihan Jin\,\orcidlink{0009-0006-6205-2203},
Gui Zhou\,\orcidlink{0000-0003-1812-8830},~\IEEEmembership{Member,~IEEE,}\\
Robert Caiming Qiu\,\orcidlink{0000-0002-0988-5525},~\IEEEmembership{Fellow,~IEEE,}
and Zenan Ling\,\orcidlink{0000-0001-8047-0253},~\IEEEmembership{Member,~IEEE}
\thanks{The work of Zenan Ling was supported in part by the National Natural Science Foundation of China (via NSFC-62406119), the Natural Science Foundation of Hubei Province (2024AFB074), the Guangdong Provincial Key Laboratory of Mathematical Foundations for Artificial Intelligence (2023B1212010001), and the Key Research and Development Program of Wuhan (2024050702030100). The work of Robert Caiming Qiu was supported in part by the National Natural Science Foundation of China (via NSFC-12141107) and the Key Research and Development Program of Wuhan (2024050702030100). \emph{(Corresponding author: Zenan Ling.)}

Fenghao Zheng, Fuhai Wang, Zihan Jin, Gui Zhou, Robert Caiming Qiu, and Zenan Ling are with the School of Electronic Information and Communications (EIC), Huazhong University of Science and Technology (HUST), Wuhan 430074, China (e-mail: zhengfenghao@hust.edu.cn; wangfuhai@hust.edu.cn; jinzihan@hust.edu.cn; gui\_zhou@hust.edu.cn; caiming@hust.edu.cn; lingzenan@hust.edu.cn).}%
}

\maketitle

\begin{abstract}
  Large-aperture reconfigurable intelligent surfaces (RISs) enable high-resolution 2D direction-of-arrival (DoA) estimation, but existing approaches still tie hardware cost and control overhead to aperture size. To decouple the effective sensing aperture from the number of physically deployed RIS elements, we present V-RIS, a framework for virtual-aperture surface-field reconstruction and DoA estimation. V-RIS uses only four corner subarrays and a single-antenna receiver to reconstruct the virtual-aperture surface field from receiver observations collected under multiple RIS phase configurations, and then performs DoA estimation on the reconstructed virtual-aperture surface field. Our key observation is that, under far-field illumination, the discretized RIS surface field satisfies finite-order spatial recurrences along both aperture axes. We enforce data-level consistency through the RIS-coded receiver observations and propagation consistency through the far-field spatial recurrence, while using a four-corner deployment geometry that retains both contiguous local elements and long aperture baselines. To improve robustness in practical receiver observations, we adopt a bias-invariant receiver-domain loss that suppresses quasi-static hardware distortions and configuration-invariant multipath contributions. Extensive simulations show that V-RIS approaches the DoA accuracy of a full-aperture benchmark while producing cleaner spectra than matrix-completion and least-squares baselines. An outdoor prototype further validates the design: with only 25\% programmable elements, V-RIS keeps both elevation and azimuth errors within $1^\circ$ of the ground truth.
\end{abstract}

\begin{IEEEkeywords}
RIS, localization, DoA, virtual aperture, surface-field reconstruction
\end{IEEEkeywords}

\section{Introduction}
Wireless localization relies on extracting spatial information from received signals for applications such as navigation~\cite{10858311}, environment awareness~\cite{Ma2024MovableSensing,10745216}, and integrated sensing and communication~\cite{10243495,11111722}. As an important component of spatial information, angular information can be obtained through 2D direction-of-arrival (DoA) estimation, which resolves the elevation and azimuth of incident targets. High-resolution DoA estimation generally benefits from a large planar aperture because angular sensitivity increases with the spatial extent and second-order moment of the array~\cite{10.1007/s11045-011-0160-5,10557757,10858311}. The large physical area and element-wise phase control of a reconfigurable intelligent surface (RIS) provide both the aperture span and configurable spatial responses required for angular sensing. By successively applying known phase configurations, the RIS projects the incident field onto a sequence of spatial codes, and an external receiver records one complex receiver observation for each code~\cite{10243495,10042005,10595504}. This coded acquisition turns the passive surface into a programmable sensing aperture and reception is performed by a single external receiver.

This sensing architecture, however, involved a fundamental \emph{cost-accuracy trade-off} in practical RIS implementations. A fully programmable RIS can fully exploit a large aperture for high-resolution angular sensing, but doing so requires element-wise phase control across the entire surface, resulting in substantial hardware complexity, signaling overhead, and calibration cost~\cite{10232975,ROSSANESE2024110208,9952197,10600711}. To reduce the hardware cost of a fully programmable RIS, a practical alternative is to deploy a sparse set of programmable RIS elements over the same large sensing aperture. Although this sparse deployment preserves the physical aperture, it provides only sparse spatial samples of the incident field, making it difficult to fully exploit the aperture for high-resolution DoA estimation. Furthermore, this challenge is exacerbated when considering a single-antenna receiver. The deployed RIS elements passively modulate and reflect the incident field, while the external receiver acquires only a configuration-dependent superposition of all reflected signals rather than the individual response of each RIS element. Consequently, the surface field is not directly observable even at the deployed RIS locations, while no measurements are available over the undeployed regions of the virtual aperture. The central challenge is therefore to recover the continuous surface field over the virtual aperture from sparse coded observations acquired by a single receiver, thereby enabling the angular resolution of a much larger fully programmable RIS with substantially lower hardware cost.

In this paper, we address the cost-accuracy trade-off in RIS-aided DoA estimation. By reconstructing a virtual-aperture surface field from a sparsely deployed RIS architecture, we propose a \underline{\textbf{v}}irtual-aperture \underline{\textbf{RIS}} (\underline{\textbf{V-RIS}}) framework for high-resolution 2D DoA estimation. As illustrated in Fig.~\ref{fig:scenario}, V-RIS physically deploys only four small programmable subarrays at the aperture corners and treats the intervening positions within the aperture as virtual. Each corner subarray captures the local field structure, while the wide separation among subarrays preserves the large aperture required for high-resolution DoA estimation. A single-antenna receiver collects multiple measurements under different RIS phase configurations, from which V-RIS reconstructs the virtual-aperture surface field before applying a conventional 2D DoA estimator. In this way, RIS reconfiguration supplies observation diversity for surface-field reconstruction, and virtual-aperture surface-field reconstruction decouples the sensing aperture from the number of physically deployed programmable elements.

\begin{figure}
  \centering
  \hspace*{0.04\linewidth}
  \subfloat[Full-aperture RIS]{
    \includegraphics[width=0.40\linewidth]{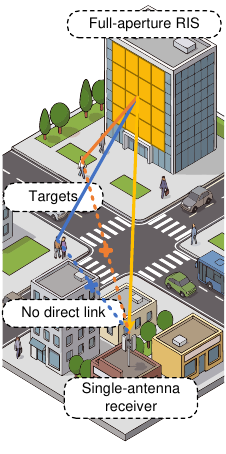}}
  \hfill
  \subfloat[Virtual-aperture RIS]{
    \includegraphics[width=0.40\linewidth]{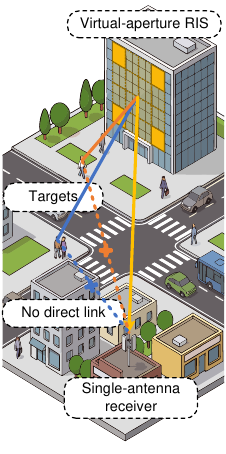}}
  \hspace*{0.04\linewidth}
  \caption{RIS-aided backward-sensing scenario without a direct target--receiver link for high-resolution 2D DoA estimation.}
  \label{fig:scenario}
\end{figure}

Recovering the virtual-aperture surface field from indirect RIS-coded observations, acquired by a single-antenna receiver, is fundamentally an ill-posed inverse problem. An implicit neural representation (INR) is employed to parameterize the continuous aperture field. We then reveal a finite-order spatial recurrence of the narrowband far-field surface field, providing a physical prior for virtual-aperture reconstruction. Building upon this insight, we further propose a dual-consistency reconstruction framework unifying data-level and propagation consistency. Specifically, data-level consistency matches the reconstructed field to the RIS-coded observations, whereas propagation consistency enforces the revealed spatial recurrence over the virtual aperture. Together, they bridge receiver-domain observations and aperture-domain propagation physics, enabling reliable virtual-aperture reconstruction from limited measurements. For robust reconstruction, a two-stage procedure first estimates the surface field from the receiver observations and then jointly refines the INR and recurrence coefficients. Finally, a closed-form affine alignment compensates for the configuration-invariant additive term and unknown global complex gain in practical systems without introducing trainable calibration variables.

The main contributions are summarized as follows:
\begin{itemize}
  \item \textit{Virtual-aperture sensing from indirect RIS observations:} We formulate virtual-aperture surface-field reconstruction for a sparsely programmed passive RIS and a single-antenna receiver that provides only one coded scalar observation per configuration. Unlike direct localization and array-completion methods, V-RIS recovers a reusable aperture-domain field without assuming element-wise or covariance samples, thereby decoupling the sensing aperture from the number of programmable elements.

  \item \textit{Virtual-aperture surface-field reconstruction and deployment co-design:} We develop a physics-informed INR framework that recovers a propagation-consistent full-aperture RIS field from sparse coded receiver observations. We further derive a four-corner deployment that jointly supports reliable field reconstruction and a large effective aperture for high-resolution 2D DoA estimation, making virtual-aperture sensing feasible under severe spatial undersampling.

  \item \textit{Simulation and prototype validation:} We validate V-RIS through numerical experiments, ablation studies, and an outdoor prototype. V-RIS recovers all simulated targets and yields the cleanest spatial spectrum, while the outdoor prototype maintains sub-degree accuracy with only $0.41^\circ$ and $0.21^\circ$ increases in elevation and azimuth errors, respectively, relative to full-aperture programming.
\end{itemize}

\section{Related Work}
The proposed V-RIS is closely related to three research directions: RIS-aided DoA estimation, sparse and virtual arrays, and physics-informed implicit neural representations (INRs). However, existing studies do not address virtual-aperture surface-field reconstruction from indirect RIS-coded single-receiver observations, which constitutes the central problem considered in this work.

\textbf{RIS-aided DoA estimation.}
Existing RIS-assisted sensing and localization methods generally estimate angular or channel parameters directly from pilot measurements collected under multiple RIS configurations~\cite{10243495,9555250,10387221,10643306,s25134140}. Grid-based sparse recovery, robust norm minimization, and gridless atomic-norm methods have been developed to improve resolution or reduce search complexity~\cite{9555250,10643306,10533449}. These methods treat the RIS response primarily as a forward model for parameter estimation. However, they do not address reconstruction of a matrix-valued virtual-aperture surface field when only a small portion of the aperture is physically programmable. V-RIS instead separates the problem into surface-field reconstruction and subsequent DoA estimation. This intermediate surface field retains the spatial phase structure required by conventional array processing and can be inspected independently of the final angular estimate. Calibration studies have also characterized phase-dependent responses, coupling, and other RIS hardware impairments~\cite{11205960,li2025risbeamcalibrationisac}. Our bias-invariant receiver-domain loss is complementary to such calibration: it eliminates global complex-gain and configuration-independent additive terms within each reconstruction iteration rather than requiring a separately calibrated forward model.

\textbf{Sparse RIS and virtual-aperture acquisition.}
Sparse arrays, antenna selection, and nonredundant layouts reduce the number of radio-frequency chains or sensors while preserving informative baselines~\cite{9443458,9962831,11205089}. RIS systems with sparse active or sensing elements similarly reduce channel-acquisition overhead by embedding a limited number of receive-capable elements in the surface~\cite{asif2023isac,10042005}. Such architectures still obtain element-wise measurements at the active sensors. Classical virtual-array methods then enlarge the effective aperture through coarray processing, covariance completion, matrix/tensor completion, or spectral compressed sensing~\cite{s22103754,10087326,10474161,10446340,10595989}. Their input is therefore direct spatial or covariance samples. In contrast, a passive RIS connected to one external receiver provides only configuration-dependent scalar projections of the deployed-element field. V-RIS reconstructs the deployed- and virtual-element field samples from these coded projections rather than assuming that the deployed-element samples are directly measured. Moreover, its four-corner layout is not introduced merely as a generic sparse aperture: each corner block supplies the consecutive samples required to estimate the spatial recurrence coefficients, while the separation between blocks preserves the aperture span.

\textbf{INR and physics-informed reconstruction.}
INRs provide continuous coordinate-based representations and have been studied for signal reconstruction and RIS codebook optimization~\cite{NEURIPS2022_575c4500,11251044}. Physics-informed learning further constrains neural fields using governing equations or electromagnetic models~\cite{schoder2024feasibility,11016226,11172275,Tang2025}. NEAR~\cite{Bu2025NEARNE} reconstructs dense radar-array responses with an untrained coordinate INR, row- and column-wise block-Hankel linear-prediction constraints, and a two-stage optimization that first fits the measured array responses and then jointly estimates the INR and prediction coefficients. V-RIS shares this INR-plus-spatial-recurrence principle, but adapts it to a different sensing interface and deployment objective. The NEAR method fits responses directly measured at selected radar-array locations, whereas a passive RIS with one external receiver provides only configuration-dependent scalar mixtures of the deployed-element field. V-RIS therefore replaces direct sample-domain fitting with a bias-invariant receiver-domain loss that enforces data-level consistency under RIS-coded receiver observations. Its four-corner layout is further designed to retain both consecutive local samples for recurrence-coefficient estimation and long baselines across the full aperture, and the reconstructed virtual-aperture surface field is evaluated through downstream 2D DoA estimation.

\section{System Model}
This section develops the V-RIS-aided DoA estimation framework illustrated in Fig.~\ref{fig_framework}. The objective is to estimate the 2D DoAs of $K$ targets from RIS-coded receiver observations reflected by several distributed RIS subarrays. We first formulate the target--RIS--Rx observation process and relate the virtual-aperture surface field to 2D DoA estimation through the Bartlett spectrum. Based on this formulation, we pose the reconstruction problem and identify three key requirements for reliable reconstruction: data-level consistency with the receiver observations, propagation consistency across the full aperture, and a low-cardinality deployment geometry that supports both.

\begin{figure*}[t]
  \centering
  \includegraphics[width=\textwidth]{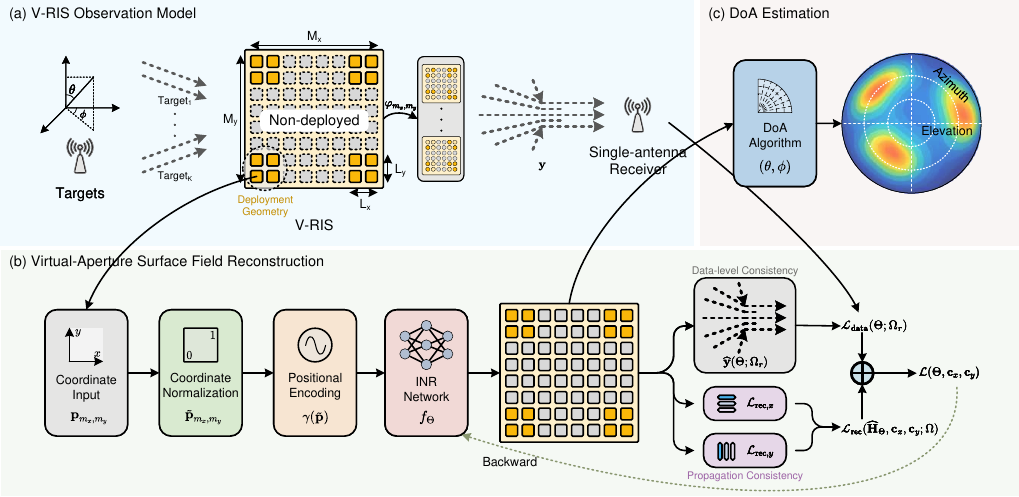}
  \caption{Overview of the proposed V-RIS DoA estimation framework. (a) RIS-coded observation of $K$ far-field targets are acquired using four corner subarrays and a single-antenna receiver. (b) INR-based virtual-aperture surface-field reconstruction with data-level and propagation consistency constraints. (c) 2D DoA estimation from the reconstructed field.}
  \label{fig_framework}
\end{figure*}

\subsection{Signal Model}
\label{subsection:preliminary}
We consider a V-RIS that defines a virtual aperture on the $xy$-plane for DoA estimation, discretized into an $M_x\times M_y$ uniform planar grid with inter-element spacings $d_x$ and $d_y$. Its complete element index set is $\Omega\triangleq\{(m_x,m_y):1\leq m_x\leq M_x,\ 1\leq m_y\leq M_y\}$, where the grid point $(m_x,m_y)\in\Omega$ is located at $\mathbf p_{m_x,m_y}=[(m_x-1)d_x,(m_y-1)d_y,0]^\top$. Thus, the full aperture contains $M=|\Omega|=M_xM_y$ grid points. $\Omega$ is divided into two parts: $\Omega_r\subset\Omega$ contains the indices of the physically deployed and programmable elements, of cardinality $M_r\triangleq|\Omega_r|$, whereas $\Omega_v\triangleq\Omega\setminus\Omega_r$ contains the indices of the virtual elements.

The sensing model is illustrated in Fig.~\ref{fig_framework}(a). Let $N$ denote the number of snapshots, indexed by $n\in\{1,\ldots,N\}$. At the $n$-th snapshot, the RIS applies a programmable phase configuration to the deployed elements in $\Omega_r$, reflecting the signals from the $K$ targets. The resulting RIS-coded observation at the single-antenna receiver is given by~\cite{10243495,10387221,11301934}
\begin{equation}
  \begin{aligned}
  y[n]
  &=\sum_{(m_x,m_y)\in\Omega_r}
  \underbrace{\mathbf G_{m_x,m_y}}_{\text{\textcircled{1} RIS--Rx}}
  \underbrace{e^{j\Phi_{m_x,m_y}[n]}}_{\text{\textcircled{2} RIS}}
  \underbrace{\mathbf H_{m_x,m_y}}_{\text{\textcircled{3} Target--RIS}}
  +\underbrace{z[n]}_{\text{\textcircled{4} Noise}},
  \end{aligned}
  \label{eq:measurement}
\end{equation}
where:
\begin{itemize}
  \item \textit{\textcircled{1} RIS--Rx:} $\mathbf G_{m_x,m_y}$ is assumed to be known due to the fixed RIS and receiver positions, allowing its geometry-dependent free-space response to be computed a priori.
  \item \textit{\textcircled{2} RIS:} $\Phi_{m_x,m_y}[n]$ is the programmed phase of element $(m_x,m_y)$ at the $n$-th snapshot.
  \item \textit{\textcircled{3} Target--RIS:} The propagation path from each target to the receiver comprises a direct link and an RIS-reflected link. The direct links are assumed to be blocked, so the receiver observes the targets only through the reflected links, for which only LoS propagation is considered. Accordingly, $\mathbf H_{m_x,m_y}$ denotes the superposition of the target--RIS LoS components from all $K$ far-field targets at RIS element $(m_x,m_y)$. $\theta_k$ and $\phi_k$ denote the elevation and azimuth angles of the $k$-th target, respectively.
  \item \textit{\textcircled{4} Noise:} $z[n]$ denotes additive noise.
\end{itemize}

The observations over the $N$ snapshots are collected into $\mathbf y=[y[1],\ldots,y[N]]^\top\in\mathbb C^{N\times1}$.

Given the surface field $\mathbf H$, the 2D Bartlett spectrum~\cite{VanTrees2002OptimumArray} is
\begin{equation}
  P(\theta,\phi)
  =\frac{\left|\left\langle\mathbf S(\theta,\phi),\mathbf H\right\rangle_{\mathrm F}\right|^2}
  {\left\|\mathbf S(\theta,\phi)\right\|_{\mathrm F}^2},
  \label{eq:bartlett_spectrum}
\end{equation}
where $\langle \cdot,\cdot\rangle_{\mathrm F}$ is the Frobenius inner product, $|\cdot|$ denotes the modulus of a complex scalar, $\|\cdot\|_{\mathrm F}$ is the Frobenius norm, and $\mathbf S(\theta,\phi)\in\mathbb C^{M_x\times M_y}$ is the steering matrix defined by $[\mathbf S(\theta,\phi)]_{m_x,m_y} = \exp(-j\frac{2\pi}{\lambda}[(m_x-1)d_x\sin\theta\cos\phi+(m_y-1)d_y\sin\theta\sin\phi])$.

Our objective is to estimate the 2D DoAs of $K$ targets by reconstructing the virtual-aperture surface field $\mathbf H$ from the sparse receiver observations $\mathbf y$, and then identifying the DoAs from the peaks of the Bartlett spectrum in \eqref{eq:bartlett_spectrum}. Since \eqref{eq:measurement} involves only the deployed elements indexed by $\Omega_r$, leaving the field entries on $\Omega_v$ unobserved. Recovering the full-aperture estimate $\widehat{\mathbf H}$ therefore requires jointly accounting for the receiver observations, the spatial structure across the full aperture, and the geometry of deployed elements. The recovery problem will be formulated in the following subsection.

\subsection{Problem Formulation}
\label{subsection:problem}
The core problem in V-RIS is to reconstruct the virtual-aperture surface field from the receiver observations while determining a low-cardinality deployment. Specifically, we jointly seek a full-aperture estimate $\widehat{\mathbf H}\in\mathbb C^{M_x\times M_y}$ and a deployed set $\Omega_r$ such that the deployed entries produce receiver-domain predictions consistent with $\mathbf y$ through \eqref{eq:measurement}, while the full-aperture estimate obeys the same far-field propagation law as the true surface field $\mathbf H$.

Each receiver observation in \eqref{eq:measurement} depends only on the $M_r$ entries indexed by $\Omega_r$, whereas the remaining entries on $\Omega_v$ do not appear in the observation model. Thus, when $M_r\ll M$, measurement consistency alone is insufficient to recover the full-aperture field. An additional physics-based constraint is required to restrict the admissible full-aperture fields and thereby link the deployed and virtual entries, while $\Omega_r$ must be chosen to support this linkage. This leads to the following joint reconstruction and deployment formulation:
\begin{subequations}
\begin{align}
\text{find}\quad
& \widehat{\mathbf H}\in\mathbb C^{M_x\times M_y},\ \Omega_r
\label{eq:compatibility_a}
\\[0.2em]
\text{s.t.}\quad
& \bigl\|\mathbf y-
  \widehat{\mathbf y}(\widehat{\mathbf H};\Omega_r)\bigr\|_2^2
\leq \varepsilon_{\mathrm{data}}
\label{eq:compatibility_b}\\
& \mathcal C_{\mathrm{phy}}(\widehat{\mathbf H};\Omega)
  \leq \varepsilon_{\mathrm{prop}}
\label{eq:compatibility_c}\\
& \Omega_r\in\mathcal D_{M_r}
\label{eq:compatibility_d}
\end{align}
\label{eq:sparse_to_virtual_compatibility}
\end{subequations}
Three requirements for the reconstruction are specified in \eqref{eq:sparse_to_virtual_compatibility}:
\begin{enumerate}
  \item [\eqref{eq:compatibility_b}] \textit{Data-level consistency:} the reconstructed field entries on $\Omega_r$ must explain the receiver observations through the target--RIS--Rx cascade in \eqref{eq:measurement}, i.e., $\mathbf y\approx\widehat{\mathbf y}(\widehat{\mathbf H};\Omega_r)$. Here, $\widehat{\mathbf y}(\widehat{\mathbf H};\Omega_r)$ is obtained from \eqref{eq:measurement} by replacing $\mathbf H_{m_x,m_y}$ with its estimate $\widehat{\mathbf H}_{m_x,m_y}$, and $\varepsilon_{\mathrm{data}}$ is the tolerance for the receiver-domain data mismatch.
  \item [\eqref{eq:compatibility_c}] \textit{Propagation consistency:} Since the true full-aperture field $\mathbf H$ is unknown, the reconstruction error is not directly accessible. We therefore express propagation consistency in \eqref{eq:compatibility_c} through a physics-based measure $\mathcal C_{\mathrm{phy}}$, which evaluates the violation of the assumed propagation model, with $\varepsilon_{\mathrm{prop}}$ specifying the admissible tolerance.
  \item [\eqref{eq:compatibility_d}] \textit{Deployment geometry:} $\mathcal D_{M_r}$ denotes the feasible set of deployment geometries with low cardinality $M_r\ll M$ that can support both the data-level consistency in \eqref{eq:compatibility_b} and the propagation consistency in \eqref{eq:compatibility_c}. A feasible deployed set $\Omega_r$ is selected offline before collecting $\mathbf y$ and remains unchanged during surface-field reconstruction.
\end{enumerate}

With the observation model and the three reconstruction requirements established, we next develop the V-RIS-aided DoA estimation method accordingly, thereby completing the overall DoA estimation pipeline.

\section{V-RIS-aided DoA estimation}
\label{section:method}
To satisfy the three requirements in \eqref{eq:sparse_to_virtual_compatibility}, we develop a dual-consistency reconstruction framework for virtual-aperture reconstruction. Data-level consistency is enforced through the INR parameterization and a bias-invariant receiver-domain loss in Section~\ref{subsection:inr_parameterization}, while propagation consistency is imposed through the finite-order spatial recurrence developed in Section~\ref{subsection:propagation_consistency}. To enable a practical low-cardinality deployment, a four-corner sampling layout is adopted in Section~\ref{subsection:sampling_geometry}. These components are integrated into a unified optimization problem and solved via the two-stage reconstruction procedure presented in Section~\ref{subsection:joint_optimization}. Finally, the reconstructed virtual-aperture surface field is used for 2D DoA estimation based on the Bartlett spectrum in Section~\ref{subsection:doa}.

\subsection{Data-Level Consistency: INR and Bias-Invariant Loss}
\label{subsection:inr_parameterization}
Enforcing data-level consistency requires the reconstructed virtual-aperture surface field to reproduce the receiver observations through the target--RIS--Rx cascade in \eqref{eq:measurement}. This requirement presents two challenges:
\begin{itemize}
  \item \textit{Surface-field representation:} the field contains $M$ complex entries, whereas $\mathbf y$ provides only indirect receiver observations through the deployed set $\Omega_r$. The complete field therefore be represented without learning every entry independently from these sparse receiver observations.
  \item \textit{Receiver-domain bias:} practical receiver observations differ from the nominal cascade response because of an unknown global complex gain and a configuration-invariant additive contribution. Direct fitting would transfer these nuisance effects to the reconstructed field.
\end{itemize}
We address the two challenges with an INR and a bias-invariant receiver-domain loss, respectively.

\textbf{INR.} The representation challenge arises because the full aperture contains $M$ complex field values, while $\mathbf y$ constrains them only through receiver observations over $\Omega_r$. Treating these values as independent variables would ignore their common coordinate domain and be difficult under sparse receiver observations. We therefore use an INR, which shares one set of parameters across the aperture and assigns a complex field value to every deployed and virtual coordinate through a coordinate-to-value mapping~\cite{NEURIPS2022_575c4500}.

We parameterize the surface field by the coordinate-to-field map $f_\Theta:[0,1]^2\rightarrow\mathbb C$, where $\Theta$ denotes the learnable parameters. To capture high-frequency spatial variations, the grid coordinate is first normalized as $\widetilde{\mathbf p}_{m_x,m_y}=[(m_x-1)/(M_x-1),(m_y-1)/(M_y-1)]^\top\in[0,1]^2$ and mapped by the positional encoding~\cite{mildenhall2020nerf}
\begin{equation}
  \begin{aligned}
  \gamma(\widetilde{\mathbf p})
  =[
  &\sin(\pi\widetilde{\mathbf p}), \cos(\pi\widetilde{\mathbf p}), \cdots, 
  \\&\sin(2^{B-1}\pi\widetilde{\mathbf p}), \cos(2^{B-1}\pi\widetilde{\mathbf p})
  ]
  ^\top \in\mathbb R^{4B},
  \end{aligned}
  \label{eq:fourier_encoding}
\end{equation}
where the sine and cosine functions are applied elementwise and $B$ denotes the number of frequency levels in the encoding. Since the field is complex-valued, we parameterize its real and imaginary parts separately:
\begin{equation}
  f_\Theta(\widetilde{\mathbf p})
  =f_{\Theta,\Re}(\gamma(\widetilde{\mathbf p}))
  +j f_{\Theta,\Im}(\gamma(\widetilde{\mathbf p})),
\end{equation}
where $f_{\Theta,\Re},f_{\Theta,\Im}:\mathbb R^{4B}\rightarrow\mathbb R$ are multilayer perceptrons (MLP), each with hidden layers of width 256 and ReLU activations. Evaluating the INR over the aperture yields
\begin{equation}
  \widehat{\mathbf H}_\Theta
  =[f_\Theta(\widetilde{\mathbf p}_{m_x,m_y})]_{m_x,m_y}
  \in\mathbb C^{M_x\times M_y}.
  \label{eq:inr_field_map}
\end{equation}

The INR surface field in \eqref{eq:inr_field_map} is mapped to the nominal receiver response through the target--RIS--Rx cascade as \eqref{eq:measurement}:
\begin{equation}
  \widehat{y}[n](\Theta;\Omega_r)
  =\sum_{(m_x,m_y)\in\Omega_r}
  \mathbf G_{m_x,m_y}\cdot e^{j\Phi_{m_x,m_y}[n]}\widehat{\mathbf H}_{\Theta,m_x,m_y}.
  \label{eq:inr_predicted_measurement}
\end{equation}

\textbf{Bias-invariant loss.} The INR provides a nominal receiver prediction through \eqref{eq:inr_predicted_measurement}, but directly fitting this prediction to practical receiver observations can distort the reconstructed field. Unknown amplitude and phase offsets produce a global complex scaling, while quasi-static coupling, residual leakage, and configuration-invariant multipath introduce an additive contribution shared across RIS configurations. 

We relate the receiver observations to the INR prediction through the affine model
\begin{equation}
  y[n]
  =\rho\widehat{y}[n](\Theta;\Omega_r)+\nu+z[n],
  \label{eq:affine_measurement_model}
\end{equation}
where $\rho\in\mathbb C$ captures the constant amplitude and phase offsets, $\nu\in\mathbb C$ accounts for the configuration-invariant additive contribution. Rather than treating $\rho$ and $\nu$ as trainable variables, we eliminate them analytically and use the minimized residual as the bias-invariant receiver-domain loss.

\begin{proposition}[Closed-form receiver-domain affine alignment]
  \label{prop:affine_alignment}
  Given the observed vector $\mathbf y$ and an INR prediction $\widehat{\mathbf y}(\Theta;\Omega_r)$, let $\mathbf P=\mathbf I_N-\mathbf 1_N\mathbf 1_N^\top/N$, $\mu_y=\mathbf 1_N^\top\mathbf y/N$, and $\mu_{\widehat y}=\mathbf 1_N^\top\widehat{\mathbf y}(\Theta;\Omega_r)/N$. If $\mathbf P\widehat{\mathbf y}(\Theta;\Omega_r)\ne\mathbf 0$, the solution to
  \begin{equation}
    \min_{\rho,\nu\in\mathbb C}\frac{1}{N}
    \left\|\mathbf y-\rho\widehat{\mathbf y}(\Theta;\Omega_r)-\nu\mathbf 1_N\right\|_2^2
  \end{equation}
  is
  \begin{equation}
    \begin{aligned}
    \rho^\star(\Theta;\Omega_r)
    &=\frac{\widehat{\mathbf y}(\Theta;\Omega_r)^H\mathbf P\mathbf y}
    {\widehat{\mathbf y}(\Theta;\Omega_r)^H\mathbf P\widehat{\mathbf y}(\Theta;\Omega_r)},\\
    \nu^\star(\Theta;\Omega_r)
    &=\mu_y-\rho^\star(\Theta;\Omega_r)\mu_{\widehat y}.
    \end{aligned}
    \label{eq:affine_alignment}
  \end{equation}
\end{proposition}

\begin{proof}
  See Appendix~\ref{Appendix:proof_LS}.
\end{proof}

The minimized value in Proposition~\ref{prop:affine_alignment} gives the bias-invariant receiver-domain loss
\begin{equation}
  \mathcal L_{\mathrm{data}}(\Theta;\Omega_r)
  =\frac{1}{N}\Bigl\|\mathbf P\mathbf y-\rho^\star(\Theta;\Omega_r)\mathbf P\widehat{\mathbf y}(\Theta;\Omega_r)\Bigr\|_2^2,
  \label{eq:BI_loss}
\end{equation}

Centering removes any configuration-invariant additive offset, while $\rho^\star(\Theta;\Omega_r)$ compensates for the unknown global complex gain mismatch between the observed and predicted samples. Consequently, the receiver observations cannot determine the common amplitude and phase of $\widehat{\mathbf H}$, because this common factor can be absorbed into $\rho$. This factor leaves the relative spatial pattern across the aperture unchanged and therefore does not alter the DoA estimates.

This receiver-domain formulation establishes data-level consistency. The next subsection introduces propagation consistency to constrain the unobserved entries on $\Omega_v$.

\subsection{Propagation Consistency: Spatial Recurrence}
\label{subsection:propagation_consistency}
The data-level consistency enforced by $\mathcal L_{\mathrm{data}}$ constrains $\widehat{\mathbf H}_\Theta$ only through the response of its entries on $\Omega_r$ and thus provides no direct constraint on $\Omega_v$. To propagate the information supported by the receiver observations from $\Omega_r$ to $\Omega_v$, the reconstruction must satisfy propagation consistency.

\begin{proposition}[Finite-order spatial recurrence under far-field condition~\cite{Bu2025NEARNE}]
  \label{prop:recurrence}
  If the target--RIS link contains $K$ far-field targets, there exist $\mathbf c_x,\mathbf c_y\in\mathbb C^{K\times1}$ such that, for valid indices:
  \begin{align}
    \mathbf H_{m_x,m_y}
    &=\sum_{k=1}^{K}c_{x,k}\mathbf H_{m_x-k,m_y},
    &&m_x>K,
    \label{eq:rec_x}\\
    \mathbf H_{m_x,m_y}
    &=\sum_{k=1}^{K}c_{y,k}\mathbf H_{m_x,m_y-k},
    &&m_y>K.
    \label{eq:rec_y}
  \end{align}
\end{proposition}

Proposition~\ref{prop:recurrence} reveals that the matrix entries $\mathbf H_{m_x,m_y}$ satisfy finite-order spatial recurrences along both aperture axes. To impose these dependencies on the reconstructed field, let $\widehat{\mathbf H}_{m_x,m_y}$ denote its value at grid point $(m_x,m_y)$ and define the recurrence residuals over all valid grid points as
\begin{equation}
  \begin{aligned}
    \mathcal L_{\mathrm{rec},x}(\widehat{\mathbf H},\mathbf c_x;\Omega)
    &=\sum_{m_y=1}^{M_y}\sum_{m_x=K+1}^{M_x}\\
    &\left|\widehat{\mathbf H}_{m_x,m_y}-\sum_{k=1}^{K}c_{x,k}\widehat{\mathbf H}_{m_x-k,m_y}\right|^2,\\
    \mathcal L_{\mathrm{rec},y}(\widehat{\mathbf H},\mathbf c_y;\Omega)
    &=\sum_{m_x=1}^{M_x}\sum_{m_y=K+1}^{M_y}\\
    &\left|\widehat{\mathbf H}_{m_x,m_y}-\sum_{k=1}^{K}c_{y,k}\widehat{\mathbf H}_{m_x,m_y-k}\right|^2.
  \end{aligned}
  \label{eq:Lphy}
\end{equation}
where $\mathbf c_x,\mathbf c_y\in\mathbb C^{K\times1}$ are learnable recurrence coefficients. The total recurrence loss is
\begin{equation}
  \mathcal L_{\mathrm{rec}}(\widehat{\mathbf H},\mathbf c_x,\mathbf c_y;\Omega)
  =\mathcal L_{\mathrm{rec},x}(\widehat{\mathbf H},\mathbf c_x;\Omega)+\mathcal L_{\mathrm{rec},y}(\widehat{\mathbf H},\mathbf c_y;\Omega).
  \label{eq:Lrec}
\end{equation}

Because the same coefficients are shared over $\Omega$, this loss couples the entries on $\Omega_r$ and $\Omega_v$ and acts as a global structural regularizer rather than a local interpolation loss. Under far-field conditions, the true surface field satisfies $\mathcal L_{\mathrm{rec}}(\mathbf H,\mathbf c_x,\mathbf c_y;\Omega)=0$. When the number of snapshots is limited, possibly even smaller than $M_r$, this regularization complements the bias-invariant receiver-domain loss.

In \eqref{eq:Lphy}, the recurrence order determines the amount of local support required by the model, so the choice of $K$ must be specified. We set $K$ to the number of synthesized targets in simulations. In deployment, candidate values of $K$ can be selected using scene knowledge or receiver-domain validation residuals evaluated on reserved RIS configurations. When an initial deployed field estimate is available, model-order selection criteria such as the Akaike information criterion (AIC) and minimum description length (MDL) can be used to estimate the effective target count~\cite{Wax1985Detection}.

Moreover, the sampling requirement imposed by the spatial recurrence is revealed by \eqref{eq:rec_x} and \eqref{eq:rec_y}. Estimating $\mathbf c_x$ or $\mathbf c_y$ requires $K+1$ field entries involved in each recurrence equation to be available from the receiver observations. Therefore, $\Omega_r$ must contain local runs of at least $K+1$ consecutive entries along both axes. The next subsection combines this observation-supported local requirement with the aperture-spread requirement for DoA sensing.

\subsection{Deployment Geometry: Four-Corner Layout}
\label{subsection:sampling_geometry}
The physical deployment determines how reliably the virtual aperture can be synthesized. To synthesize the virtual aperture from a low-cardinality deployment, the deployment geometry must satisfy two requirements:
\begin{itemize}
  \item \textit{Global aperture:} preserve long spatial baselines along both axes for DoA resolution;
  \item \textit{Local recurrence:} provide locally consecutive samples along both axes for estimating $\mathbf c_x$ and $\mathbf c_y$.
\end{itemize}
Figure~\ref{fig:four corner} illustrates how these two requirements jointly constrain the deployment to the four-corner layout.

\begin{figure}
  \centering
  \includegraphics[width=0.9\linewidth]{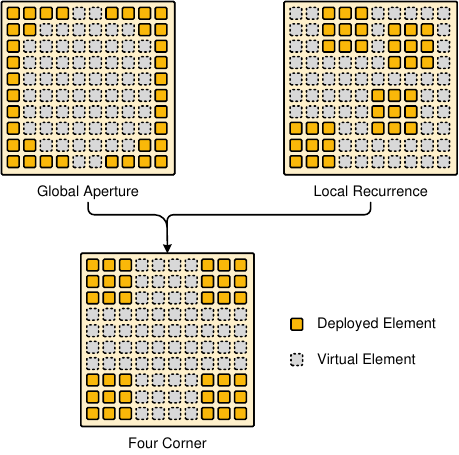}
  \caption{Global aperture favors boundary elements to preserve long baselines, whereas local recurrence requires contiguous sampling along both axes, jointly motivating the deployed of four contiguous corner subarrays.}
  \label{fig:four corner}
\end{figure}

\textbf{Global aperture.} We characterize the long-baseline requirement through the direction-cosine CRBs, which relate the spatial distribution of the deployed coordinates to the available DoA information along the two aperture axes.

\begin{proposition}[Geometry-dependent direction-cosine CRBs for low-cardinality deployments~\cite{Ma2024MovableSensing}]
  \label{prop:deployment_geometry}
  For a low-cardinality deployed set $\Omega_r$ with $M_r=|\Omega_r|$, let $\mathbf x_{\Omega_r}$ and $\mathbf y_{\Omega_r}$ collect its physical $x$- and $y$-coordinates, and let $\mathbf B_{\Omega_r}=M_r^{-1}\mathbf I_{M_r}-M_r^{-2}\mathbf 1_{M_r}\mathbf 1_{M_r}^\top$. The layout-dependent terms of the UPA direction-cosine CRBs are
  \begin{equation}
    \begin{aligned}
    \mathrm{CRB}_{u}(\Omega_r)
    &\propto
    \left[
    \mathbf x_{\Omega_r}^{\top}\mathbf B_{\Omega_r}\mathbf x_{\Omega_r}
    -\frac{(\mathbf x_{\Omega_r}^{\top}\mathbf B_{\Omega_r}\mathbf y_{\Omega_r})^2}
    {\mathbf y_{\Omega_r}^{\top}\mathbf B_{\Omega_r}\mathbf y_{\Omega_r}}
    \right]^{-1},\\
    \mathrm{CRB}_{v}(\Omega_r)
    &\propto
    \left[
    \mathbf y_{\Omega_r}^{\top}\mathbf B_{\Omega_r}\mathbf y_{\Omega_r}
    -\frac{(\mathbf x_{\Omega_r}^{\top}\mathbf B_{\Omega_r}\mathbf y_{\Omega_r})^2}
    {\mathbf x_{\Omega_r}^{\top}\mathbf B_{\Omega_r}\mathbf x_{\Omega_r}}
    \right]^{-1}.
    \end{aligned}
    \label{eq:geometry_information}
  \end{equation}
\end{proposition}

For a given deployment cardinality $M_r$, \eqref{eq:geometry_information} shows that reducing both CRBs requires enlarging the two centered coordinate spreads while suppressing the centered cross term. If $\Omega_r$ is symmetric about both aperture axes, then $\mathbf x_{\Omega_r}^{\top}\mathbf B_{\Omega_r}\mathbf y_{\Omega_r}=0$. Within this symmetric class, placing support near the extreme coordinates of only one axis enlarges only the corresponding spread. Since the same $M_r$ elements must enlarge both spreads, they should lie near the extreme coordinates of both axes simultaneously. The global CRB requirement therefore confines the support to the four corner regions of the aperture.

\textbf{Local recurrence.} The recurrence model then determines how these corner regions must be populated. As established by \eqref{eq:rec_x} and \eqref{eq:rec_y}, estimating the unknown recurrence coefficients requires data-supported runs of at least $K+1$ consecutive entries along both axes. Hence, the corner support cannot consist of isolated elements. Each corner region must form a contiguous subarray with $L_x>K$ and $L_y>K$.

\textbf{Four-corner layout.} Combining this local condition with the symmetric corner regions dictated by \eqref{eq:geometry_information} specifies four $L_x\times L_y$ corner subarrays. For $\zeta\in\{x,y\}$, define the corresponding edge bands and their Cartesian product as
\begin{equation}
  \begin{aligned}
  \mathcal E_{\zeta}(L_{\zeta})
  &=\{1,\ldots,L_{\zeta}\}\cup
  \{M_{\zeta}-L_{\zeta}+1,\ldots,M_{\zeta}\},\\
  \Omega_r
  &=\mathcal E_x(L_x)\times\mathcal E_y(L_y).
  \end{aligned}
  \label{eq:subarray_sampling}
\end{equation}
The conditions $2L_x<M_x$ and $2L_y<M_y$ keep the edge bands disjoint and give $M_r=4L_xL_y$. In \eqref{eq:subarray_sampling}, the Cartesian product of the symmetric $x$- and $y$-edge bands places every deployed element near the extreme coordinates of both axes, while $L_x>K$ and $L_y>K$ ensure that each corner block contains the field entry and its $K$ shifted predecessors required by \eqref{eq:rec_x} and \eqref{eq:rec_y}. Thus, the CRB geometry and spatial-recurrence requirements jointly determine the four-corner layout.

\subsection{Virtual-Aperture Reconstruction Objective and Procedure}
\label{subsection:joint_optimization}
The complete reconstruction flow in Fig.~\ref{fig_framework}(b) integrates the INR parameterization with the data-level and propagation-consistency mechanisms. Together with the four-corner layout, these components jointly realize the three reconstruction requirements in \eqref{eq:sparse_to_virtual_compatibility}. We formulate the resulting reconstruction as an optimization problem that minimizes the receiver-domain data error and the residual of the recurrence constraint:
\begin{equation}
  \min_{\Theta,\mathbf c_x,\mathbf c_y}\quad
  \mathcal L_{\mathrm{data}}(\Theta;\Omega_r)+\lambda_{\mathrm{rec}}\mathcal L_{\mathrm{rec}}(\widehat{\mathbf H}_\Theta,\mathbf c_x,\mathbf c_y;\Omega),
  \label{eq:main_objective}
\end{equation}
where $\lambda_{\mathrm{rec}}$ weights the recurrence constraint. In \eqref{eq:main_objective}, $\mathcal L_{\mathrm{data}}$ enforces data-level consistency through the target--RIS--Rx response, while $\mathcal L_{\mathrm{rec}}$ instantiates the generic physical-consistency measure $\mathcal C_{\mathrm{phy}}$ in \eqref{eq:compatibility_c} using the spatial recurrence derived in Section~\ref{subsection:propagation_consistency}. The deployed set $\Omega_r$ is specified by the four-corner geometry in \eqref{eq:subarray_sampling}. Thus, the optimization makes each requirement in \eqref{eq:sparse_to_virtual_compatibility} explicit.

We adopt a two-stage training strategy to improve stability and prevent inaccurate early reconstructions from dominating the recurrence-coefficient estimation.
\begin{itemize}
  \item \textit{Stage-1 (warm-up):} We set $\lambda_{\mathrm{rec}}=0$ and minimize $\mathcal L_{\mathrm{data}}(\Theta;\Omega_r)$ for $I_1$ iterations, yielding an initialization of $\Theta$ consistent with the receiver observations.
  \item \textit{Stage-2 (recurrence-constrained refinement):} Initialized with the Stage-1 solution, we set $\lambda_{\mathrm{rec}}=0.5$. 
  We optimize $\Theta$, $\mathbf c_x$, and $\mathbf c_y$ by minimizing \eqref{eq:main_objective}. This stage ensures that the virtual-aperture surface field $\widehat{\mathbf H}_\Theta$ satisfies the recurrence constraint.
\end{itemize}

Algorithm~\ref{alg:inr_recon} summarizes the reconstruction procedure.

\begin{algorithm}[t]
  \caption{Recurrence-Constrained INR for Surface Field Reconstruction}
  \label{alg:inr_recon}
  \begin{algorithmic}[1]
    \REQUIRE Receiver observations $\mathbf y$, deployed set $\Omega_r$, RIS--Rx coefficients $\{\mathbf G_{m_x,m_y}\}$, RIS phase configurations $\{\Phi_{m_x,m_y}[n]\}$, grid coordinates $\{\widetilde{\mathbf p}_{m_x,m_y}\}$, number of Fourier frequency levels $B$, recurrence order $K$, and numbers of iterations $I_1$ and $I_2$.
    \ENSURE Virtual-aperture surface field $\widehat{\mathbf H}\in\mathbb C^{M_x\times M_y}$.
    \STATE Initialize INR parameters $\Theta$ and recurrence coefficients $\mathbf c_x,\mathbf c_y\in\mathbb C^{K\times1}$.
    \STATE \textbf{Stage 1 (warm-up):} set $\lambda_{\mathrm{rec}}=0$.
    \FOR{$i=1$ to $I_1$}
      \STATE Evaluate $\widehat{\mathbf H}_\Theta=[f_\Theta(\widetilde{\mathbf p}_{m_x,m_y})]_{m_x,m_y}$.
      \STATE Compute $\widehat{\mathbf y}(\Theta;\Omega_r)$ via \eqref{eq:inr_predicted_measurement}.
      \STATE Compute $\mathcal L_{\mathrm{data}}(\Theta;\Omega_r)$ via \eqref{eq:BI_loss}.
      \STATE Update $\Theta$ by minimizing $\mathcal L_{\mathrm{data}}(\Theta;\Omega_r)$.
    \ENDFOR
    \STATE \textbf{Stage 2 (recurrence-constrained refinement):} set $\lambda_{\mathrm{rec}}=0.5$.
    \FOR{$i=1$ to $I_2$}
      \STATE Evaluate $\widehat{\mathbf H}_\Theta=[f_\Theta(\widetilde{\mathbf p}_{m_x,m_y})]_{m_x,m_y}$.
      \STATE Compute $\widehat{\mathbf y}(\Theta;\Omega_r)$ via \eqref{eq:inr_predicted_measurement}.
      \STATE Compute $\mathcal L_{\mathrm{data}}(\Theta;\Omega_r)$.
      \STATE Compute $\mathcal L_{\mathrm{rec}}(\widehat{\mathbf H}_\Theta,\mathbf c_x,\mathbf c_y;\Omega)$.
      \STATE Compute total loss via \eqref{eq:main_objective}.
      \STATE Update $\Theta,\mathbf c_x,\mathbf c_y$ jointly by minimizing \eqref{eq:main_objective}.
    \ENDFOR
    \STATE \textbf{return} $\widehat{\mathbf H}_\Theta$.
  \end{algorithmic}
\end{algorithm}

\subsection{DoA Estimation from Virtual-Aperture Surface Field}
\label{subsection:doa}
As shown in Fig.~\ref{fig_framework}(c), the virtual-aperture surface field $\widehat{\mathbf H}_\Theta$ obtained from Algorithm~\ref{alg:inr_recon} is used for DoA estimation. Specifically, $\widehat{\mathbf H}_\Theta$ is treated as the surface field of a fully deployed aperture. The DoA information is encoded in the spatial phase variation of the surface field. We employ 2D Bartlett spectral estimation in \eqref{eq:bartlett_spectrum} to $\widehat{\mathbf H}_\Theta$, and obtain the DoA estimates $\{(\widehat\theta_k,\widehat\phi_k)\}_{k=1}^{K}$ through peak searching over the resulting spatial spectrum.

\section{Numerical Validation}
This section evaluates V-RIS in controlled simulations. We first define the common simulation protocol and evaluation metrics. We then examine four questions in sequence: whether the proposed reconstruction improves DoA estimation over the ablations and benchmarks, why the four-corner layout is preferred, how the reconstruction scales with aperture size and measurement budget, and how robust it is to non-ideal RIS coefficients. All experiments follow the observation model in \eqref{eq:measurement} and Algorithm~\ref{alg:inr_recon}. 

\subsection{Evaluation Protocol}
\label{subsec:simulation}
\paragraph{Common Simulation Setting}
As shown in Fig.~\ref{fig:numerical_scenario}, we synthesize $K=3$ narrowband far-field targets at $30\mathrm{GHz}$, with wavelength $\lambda=1\mathrm{cm}$. Their DoAs are fixed at
\begin{equation}
  (\theta_k,\phi_k)\in\{(60^\circ,10^\circ),(60^\circ,80^\circ),(35^\circ,45^\circ)\}.
  \label{eq:sim_doa_components}
\end{equation}

The full aperture is represented on an $M_x\times M_y$ reference grid with $M_x=M_y=64$ and element spacing $d_x=d_y=\lambda/2$. $L\times L$ corner subarrays are physically deployed and programmable, where $L=L_x=L_y$, yielding $M_r=4L^2$. The deployment ratio is defined as $\eta = {4L^2}/{M_xM_y}$. Since $L$ is integer-valued and the realizable deployment ratios are discrete, a prescribed target deployment ratio may not always be achieved exactly. In such cases, we use the feasible $L$ that gives the closest realizable deployment ratio.

The receiver is located at $\mathbf{p}_u = [0,0,1]\mathrm{m}$, which lies in the near field of the aperture. For RIS configuration $n$, V-RIS applies a random 1-bit phase configuration $\Phi_{m_x,m_y}[n]\in\{0,\pi\}$ at each $(m_x,m_y)\in\Omega_r$. The receiver records the corresponding observation $y[n]$ according to \eqref{eq:measurement}. The target gains are assumed constant over the $N$ RIS configurations, and additive white Gaussian noise is added at the receiver.

\begin{figure}
  \centering
  \includegraphics[width=\linewidth]{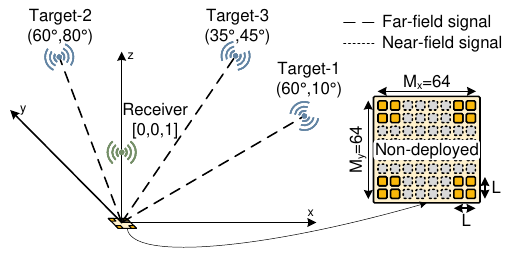}
  \caption{Numerical validation setup. The aperture is represented on a $64\times64$ grid, while only four $L\times L$ corner subarrays are deployed. The receiver is located at $\mathbf{p}_u=[0,0,1]\mathrm{m}$.}
  \label{fig:numerical_scenario}
\end{figure}

\paragraph{Reconstruction and DoA Metrics}
We evaluate both the reconstructed virtual-aperture surface field and its downstream angular information. Surface-field reconstruction is measured by the normalized mean-squared error (NMSE), reported in dB:
\begin{equation}
    \mathrm{NMSE}
    =\frac{\|\widehat{\mathbf H}-\mathbf H\|_F^2}{\|\mathbf H\|_F^2},
    \mathrm{NMSE}_{\mathrm{dB}}
    =10\log_{10}\!\left(\mathrm{NMSE}\right),
  \label{eq:nmse_def}
\end{equation}
where $\|\cdot\|_F$ denotes the Frobenius norm.

Given $\widehat{\mathbf H}$, we compute the Bartlett spectrum in \eqref{eq:bartlett_spectrum} on a $0.01^\circ$ angular grid. The $K$ strongest peaks are matched to the ground-truth targets by minimum-cost assignment. The spectra are additionally used to assess peak localization, sidelobe leakage, and spurious responses.

\subsection{DoA Estimation Performance Comparisons}
This subsection compares V-RIS with representative baselines from three method classes: sparse-aperture completion, direct RIS-aided localization, and INR-based array-response reconstruction. All spectra use the Bartlett estimator in \eqref{eq:bartlett_spectrum}.
\label{subsec:ablation_benchmark}

\paragraph{Experiment}
Unless otherwise stated, we fix $\mathrm{SNR}=20$~dB, $N=200$, and $\eta=25\%$. We use the following labels:
\begin{itemize}
  \item \textit{SHGD-O:} a sparse-aperture completion baseline. SHGD ~\cite{10474161} is given oracle corner field samples and completes the missing aperture.
  \item \textit{Direct-LS:} a direct-localization baseline. LS-based DoA recovery ~\cite{11301934} localizes targets directly from multicode receiver observations.
  \item \textit{RIS-NEAR:} an INR reconstruction baseline. NEAR~\cite{Bu2025NEARNE} is adapted from direct radar-array samples to RIS-coded receiver observations.

  \item \textit{Center-32:} an aperture-size reference using a compact $32\times32$ center aperture directly for DoA estimation.
  \item \textit{V-RIS:} the proposed virtual-aperture reconstruction.
\end{itemize}

\paragraph{Results and Analysis}
Table~\ref{tab:num_doa} reports the matched peak locations, and Fig.~\ref{fig:algorithm_polar_compare} shows the corresponding spectra. The full-aperture oracle is included only as a spectral reference for the ideal field.

\begin{table}
  \centering
  \caption{Comparison of the matched 2D DoA estimates obtained by the benchmark methods at $\mathrm{SNR}=20$~dB with $N=200$ measurements and a deployment ratio of $\eta=25\%$. V-RIS accurately recovers all targets and provides the closest estimates among the virtual-aperture reconstruction methods.}
  \label{tab:num_doa}
  \resizebox{\columnwidth}{!}{%
  \begin{tabular}{lccc}
    \hline
    & \textbf{Target 1} & \textbf{Target 2} & \textbf{Target 3} \\
    \hline
    Ground Truth & $(60.00^\circ,10.00^\circ)$ & $(60.00^\circ,80.00^\circ)$ & $(35.00^\circ,45.00^\circ)$ \\
    SHGD-O & $(64.03^\circ,7.09^\circ)$ & $(64.03^\circ,82.91^\circ)$ & $(39.00^\circ,45.00^\circ)$ \\
    Direct-LS & $(60.02^\circ,9.98^\circ)$ & $(60.04^\circ,80.02^\circ)$ & $(34.95^\circ,45.10^\circ)$ \\
    RIS-NEAR & $(60.03^\circ,10.08^\circ)$ & $(59.90^\circ,79.79^\circ)$ & $(65.40^\circ,80.26^\circ)$ \\
    Center-32 & $(59.99^\circ,10.06^\circ)$ & $(59.96^\circ,80.00^\circ)$ & $(34.98^\circ,45.05^\circ)$ \\
    V-RIS & \textbf{\boldmath$(60.01^\circ,10.00^\circ)$} & \textbf{\boldmath$(60.01^\circ,80.00^\circ)$} & \textbf{\boldmath$(35.01^\circ,45.01^\circ)$} \\
    \hline
  \end{tabular}}
\end{table}

\begin{figure}
  \centering
  \subfloat[Oracle field\label{fig:ablation-true}]{
    \includegraphics[width=0.12\textwidth]{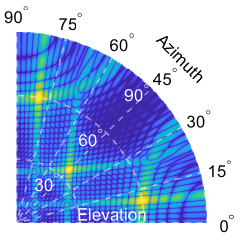}}\hfill
  \subfloat[SHGD-O\label{fig:compare-shgd}]{
    \includegraphics[width=0.12\textwidth]{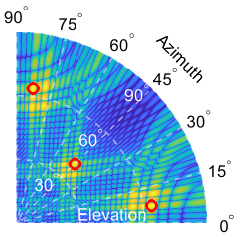}}\hfill
  \subfloat[Direct-LS\label{fig:compare-ls}]{
    \includegraphics[width=0.12\textwidth]{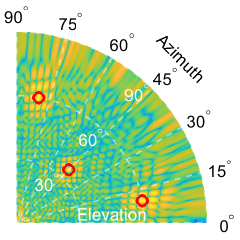}}\hfill\\
  \subfloat[RIS-NEAR\label{fig:compare-near}]{
    \includegraphics[width=0.12\textwidth]{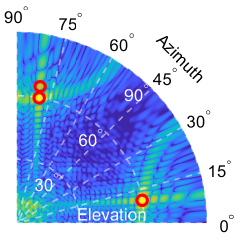}}\hfill
  \subfloat[Center-32\label{fig:compare-centre}]{
    \includegraphics[width=0.12\textwidth]{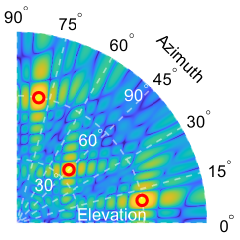}}\hfill
  \subfloat[V-RIS\label{fig:compare-ours}]{
    \includegraphics[width=0.12\textwidth]{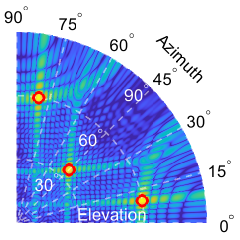}}\hfill
  \caption{Comparison of the 2D Bartlett spectra obtained from the oracle field and the benchmark methods at $\mathrm{SNR}=20$~dB with $N=200$ measurements and $\eta=25\%$, where the red circles indicate the detected peaks. V-RIS resolves all three targets and produces the cleanest non-oracle spectrum.}
  \label{fig:algorithm_polar_compare}
\end{figure}

The three method baselines show different limitations. SHGD-O recovers the three-peak structure, but its elevation errors are about $4^\circ$ and its azimuth errors reach $2.91^\circ$, indicating that sparse-aperture completion from corner samples remains biased. Direct-LS accurately localizes all targets, showing that direct localization can exploit the multicode receiver observations effectively, but it does not recover a reusable virtual-aperture surface field. RIS-NEAR, the INR-based reconstruction baseline, correctly localizes the first two targets but produces a spurious third peak near $(65^\circ,80^\circ)$.

Center-32 gives accurate point estimates, but its spectrum is broader because the physical aperture is smaller. V-RIS successfully resolves all three targets and yields the cleanest non-oracle spectrum in Fig.~\ref{fig:algorithm_polar_compare}. The comparison therefore highlights the benefit of reconstructing a large virtual-aperture surface field from sparse RIS-coded receiver observations.

\subsection{Deployment Layouts for Virtual-Aperture Reconstruction}
\label{subsec:layout_ablation}
\paragraph{Experiment}
To examine the deployment geometry, Fig.~\ref{fig:doa_polar_ablation} compares four layouts at $\mathrm{SNR}=20$~dB and $N=200$. Each sparse pattern contains 1024 deployed elements, corresponding to $\eta=25\%$ of the $64\times64$ reference aperture:
\begin{itemize}
  \item \textit{Random:} 1024 elements are randomly selected over the $64\times64$ aperture and processed by V-RIS.
  \item \textit{Center-V:} a contiguous $32\times32$ center block is virtualized to $64\times64$ by V-RIS.
  \item \textit{RIS-NEAR:} the separable sub-Nyquist deployment used by the adapted NEAR baseline.
  \item \textit{V-RIS:} the proposed four-corner deployment.
\end{itemize}
The four spectra in Fig.~\ref{fig:doa_polar_ablation} are shown without peak markers so that the comparison emphasizes the spatial sampling patterns.

\begin{figure}[t]
  \centering
  \makebox[\linewidth][c]{%
    \subfloat[Random\label{fig:ablation-random}]{%
      \includegraphics[width=0.12\textwidth]{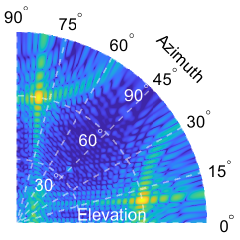}}\hfill
    \subfloat[Center-V\label{fig:ablation-centre}]{%
      \includegraphics[width=0.12\textwidth]{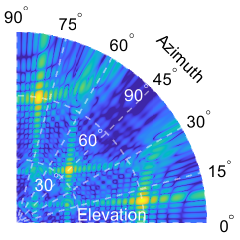}}\hfill
    \subfloat[RIS-NEAR\label{fig:ablation-near}]{%
      \includegraphics[width=0.12\textwidth]{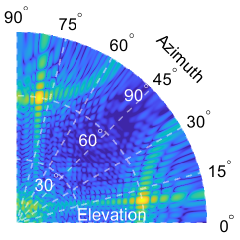}}\hfill
    \subfloat[V-RIS\label{fig:ablation-vris}]{%
      \includegraphics[width=0.12\textwidth]{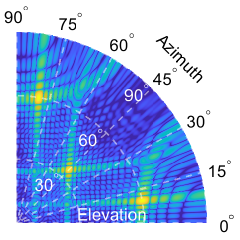}}
  }
  \caption{Comparison of the reconstructed 2D Bartlett spectra across different sparse deployment geometries, each using 1024 physical elements at $\mathrm{SNR}=20$~dB with $N=200$ measurements. The four-corner deployment most clearly preserves the three-target structure.}
  \label{fig:doa_polar_ablation}
\end{figure}

\paragraph{Results and Analysis}
Figure~\ref{fig:doa_polar_ablation} shows that equal element counts do not lead to equal reconstruction quality. Random covers the aperture but provides little contiguous local structure, so the recurrence coefficients are weakly anchored and the spectrum contains stronger background leakage. Center-V provides contiguous samples, but it does not retain the aperture extremes. These observations agree with the CRBs in \eqref{eq:geometry_information}: preserving information along both axes requires large centered coordinate spreads, while a small centered cross term avoids coupling between the two directions.

RIS-NEAR occupies a middle ground: it preserves long baselines, but its separable noncontiguous deployment is less compatible with the consecutive-sample requirement of the recurrence model. The proposed four-corner layout is the only tested layout that jointly provides the coordinate support and contiguous local neighborhoods.

\subsection{Ablation of Consistency Terms and Sparse Deployment}
\label{subsec:consistency_ablation}
\paragraph{Experiment}
To isolate the contributions of propagation consistency, the bias-invariant data term, and sparse deployment, we compare three controlled variants with the proposed method using the same numerical setting and $N=200$ measurement snapshots:
\begin{itemize}
  \item \textit{V-RIS w/o Bias:} retains propagation consistency but replaces the bias-invariant receiver-domain loss with a direct unaligned data loss.
  \item \textit{V-RIS w/o Prop.:} removes propagation consistency by setting $\lambda_{\mathrm{rec}}=0$ while retaining the data-level consistency term and four-corner deployment.
  \item \textit{V-RIS w/ Full Dep.:} deploys all $64\times64$ RIS elements while keeping the V-RIS reconstruction objective and the same $N=200$ measurement snapshots.
  \item \textit{V-RIS:} the proposed method with both consistency terms and the four-corner deployment.
\end{itemize}

\begin{table}
  \centering
  \caption{Ablation study of the bias-invariant data term, propagation consistency, and sparse deployment using $N=200$ measurement snapshots. Only the complete V-RIS configuration accurately recovers all three targets, demonstrating that both consistency terms are necessary.}
  \label{tab:ablation}
  \resizebox{\columnwidth}{!}{%
  \begin{tabular}{lccc}
    \hline
    & \textbf{Target 1} & \textbf{Target 2} & \textbf{Target 3} \\
    \hline
    Ground Truth & $(60.00^\circ,10.00^\circ)$ & $(60.00^\circ,80.00^\circ)$ & $(35.00^\circ,45.00^\circ)$ \\
    V-RIS w/o Bias & $(59.81^\circ,9.87^\circ)$ & $(77.47^\circ,53.85^\circ)$ & $(31.55^\circ,17.07^\circ)$ \\
    V-RIS w/o Prop. & $(0.00^\circ,0.00^\circ)$ & $(0.00^\circ,10.01^\circ)$ & $(30.65^\circ,0.00^\circ)$ \\
    V-RIS w/ Full Dep. & $(9.40^\circ,44.19^\circ)$ & $(8.22^\circ,58.75^\circ)$ & $(77.85^\circ,2.56^\circ)$ \\
    V-RIS & \textbf{\boldmath$(60.01^\circ,9.97^\circ)$} & \textbf{\boldmath$(60.13^\circ,79.97^\circ)$} & \textbf{\boldmath$(34.95^\circ,45.02^\circ)$} \\
    \hline
  \end{tabular}}
\end{table}

\paragraph{Results and Analysis}
The three controlled variants expose different failure modes. Without propagation consistency, the receiver-domain term can be reduced without preserving the spatial phase relation across the aperture. Removing the bias-invariant loss leaves the propagation prior intact, but the direct unaligned mismatch forces the INR to explain receiver gain and configuration-invariant additive effects through the field itself.

The full-deployment ablation shows that deploying more elements does not necessarily improve performance when the measurement-snapshot budget is fixed. The data-level inverse problem becomes more underdetermined and the phase configurations provide insufficient diversity to separate the full-deployment responses. In contrast, V-RIS combines the smaller deployed set with propagation consistency to reconstruct the virtual-aperture surface field and recovers all three targets with coordinate errors below $0.2^\circ$.

\subsection{Scalability of Virtual-Aperture Reconstruction}
\paragraph{Experiment}
We vary the virtual-aperture side length, deployment ratio, and number of RIS configurations to characterize the measurement-budget requirement of the proposed reconstruction. Four square apertures with $M_x=M_y\in\{48,64,96,128\}$ are evaluated under target deployment ratios $\eta\in\{6.25\%,12.5\%,25\%\}$ at $\mathrm{SNR}=20$~dB. 

\begin{figure}
  \centering
  \subfloat[$48\times48$\label{fig:G48}]{%
    \includegraphics[width=0.48\columnwidth]{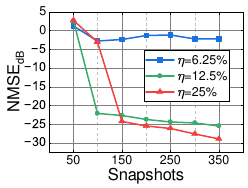}}\hfill
  \subfloat[$64\times64$\label{fig:G64}]{%
    \includegraphics[width=0.48\columnwidth]{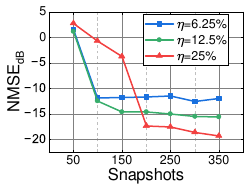}}\\
  \subfloat[$96\times96$\label{fig:G96}]{%
    \includegraphics[width=0.48\columnwidth]{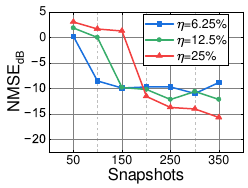}}\hfill
  \subfloat[$128\times128$\label{fig:G128}]{%
    \includegraphics[width=0.48\columnwidth]{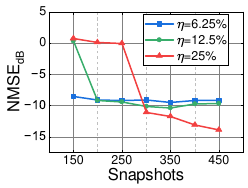}}
  \caption{Scaling behavior versus aperture size, deployment ratio, and measurement budget. Panels show NMSE versus the number of RIS configurations for $48\times48$, $64\times64$, $96\times96$, and $128\times128$ apertures at $\mathrm{SNR}=20$~dB with $\eta\in\{6.25\%,12.5\%,25\%\}$. The NMSE decreases rapidly after a measurement-budget threshold and improves with larger corner subarrays.}
  \label{fig:scaling_law}
\end{figure}

\paragraph{Results and Analysis}
Figure~\ref{fig:scaling_law} shows a clear threshold effect with respect to the measurement budget. When $N$ is small, the receiver observations are insufficient to constrain the deployed-element field and the recurrence coefficients, so the NMSE stays high. Once enough RIS configurations are available, the NMSE drops rapidly and then saturates.

For a fixed aperture, increasing $\eta$ usually lowers the final NMSE by increasing the corner size $L$. Larger corner blocks provide more consecutive samples for estimating the recurrence and give stronger anchors for surface-field reconstruction. The curves are not fully monotonic at small $N$, because a larger deployed-element set also requires more receiver observations to separate its field samples reliably. The clearest exception is the $48\times48$ aperture with $\eta=6.25\%$: $L=6$ provides too few local samples for the recurrence order $K=3$, so its NMSE remains high even when more RIS configurations are added. Overall, the results show that the proposed reconstruction can scale to larger RIS apertures, provided that the measurement budget and corner-block size are increased accordingly.

\subsection{Robustness to Non-Ideal RIS Reflection}
\paragraph{Experiment}
Practical RIS elements often exhibit amplitude attenuation and phase deviations due to finite-resolution control, component mismatch, and calibration imperfections. These non-ideal reflections cause the actual RIS response to deviate from the ideal forward model used for reconstruction, which can degrade both surface-field reconstruction and downstream DoA estimation~\cite{9534477,9115725,cao2024risinsufficientphaseshifting}. This motivates a controlled robustness test that quantifies the sensitivity of V-RIS to reflection-induced forward-model mismatch.

To emulate such non-ideal reflections, the effective reflection coefficient of element $(m_x,m_y)$ under RIS configuration $n$ is modeled as
\begin{equation}
  \tilde{\Gamma}_{m_x,m_y}[n]
  = \big(1+\epsilon_{m_x,m_y}\big)
  e^{j(\Phi_{m_x,m_y}[n]+\delta_{m_x,m_y})},
  \label{eq:amp_phase_error_model}
\end{equation}
where $\epsilon_{m_x,m_y}$ and $\delta_{m_x,m_y}$ represent the amplitude and phase perturbations in the effective reflection coefficient, respectively. To isolate the two error sources, we conduct separate amplitude and phase sweeps. Let $\mathcal{U}[a,b]$ denote the uniform distribution over $[a,b]$. For amplitude robustness, we set $\delta_{m_x,m_y}=0$ and draw $\epsilon_{m_x,m_y}$ from $\mathcal{U}[-\Delta\epsilon,0]$, with $\Delta\epsilon\in\{0,2\%,5\%,10\%,15\%,20\%\}$. For phase robustness, we set $\epsilon_{m_x,m_y}=0$ and draw $\delta_{m_x,m_y}$ from $\mathcal{U}[-\Delta\delta,\Delta\delta]$, with $\Delta\delta\in\{0^\circ,2^\circ,5^\circ,10^\circ,15^\circ,20^\circ\}$. In each Monte Carlo trial, the active perturbation is independently drawn for each deployed element and kept fixed over all RIS configurations.

We fix $\mathrm{SNR}=20$~dB, $N=200$, and $\eta=25\%$, with 50 Monte Carlo trials per impairment setting.

\begin{figure}
  \centering
  \subfloat[Amplitude perturbation\label{fig:amp_error}]{%
    \includegraphics[width=0.48\columnwidth]{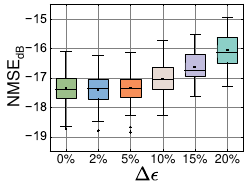}}\hfill
  \subfloat[Phase perturbation\label{fig:phase_error}]{%
    \includegraphics[width=0.48\columnwidth]{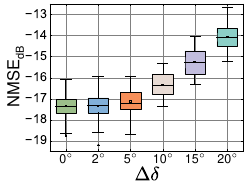}}
  \caption{Evaluation of the reconstruction NMSE under independently controlled RIS amplitude and phase perturbations at $\mathrm{SNR}=20$~dB, $N=200$, and $\eta=25\%$ over 50 Monte Carlo trials. The reconstruction remains stable over the tested perturbation ranges, although phase errors induce slightly larger NMSE variations than amplitude errors.}
  \label{fig:amp_phase_error}
\end{figure}

\paragraph{Results and Analysis}
As Fig.~\ref{fig:amp_phase_error} shows, increasing either perturbation bound raises the median NMSE. The degradation is more pronounced for phase errors than for amplitude errors over the tested range. This difference is consistent with coherent aperture processing: phase perturbations directly distort the relative spatial phase used by both recurrence-based reconstruction and subsequent spectrum formation. The gradual degradation indicates that the recurrence loss continues to constrain the field under moderate forward-model mismatch induced by reflection errors.

\section{Proof-of-Concept Prototype}
We evaluate V-RIS with an outdoor proof-of-concept prototype. This section describes the hardware, measurement configurations, reconstruction procedure, and evaluation metrics, and reports the DoA result.

\subsection{Hardware and Measurement Setup}
The prototype implements a $16 \times 16$ RIS on the $xy$-plane with element spacings $d_x = d_y = \lambda/2 \approx 2.5$~cm at carrier frequency $f_c = 5.8$~GHz, the operating frequency of the fabricated RIS. The RIS operates with 1-bit phase control, as shown in Fig.~\ref{fig:ris}.

\begin{figure}
  \centering
  \subfloat[Front view\label{fig:ris_a}]{%
    \includegraphics[width=0.48\columnwidth]{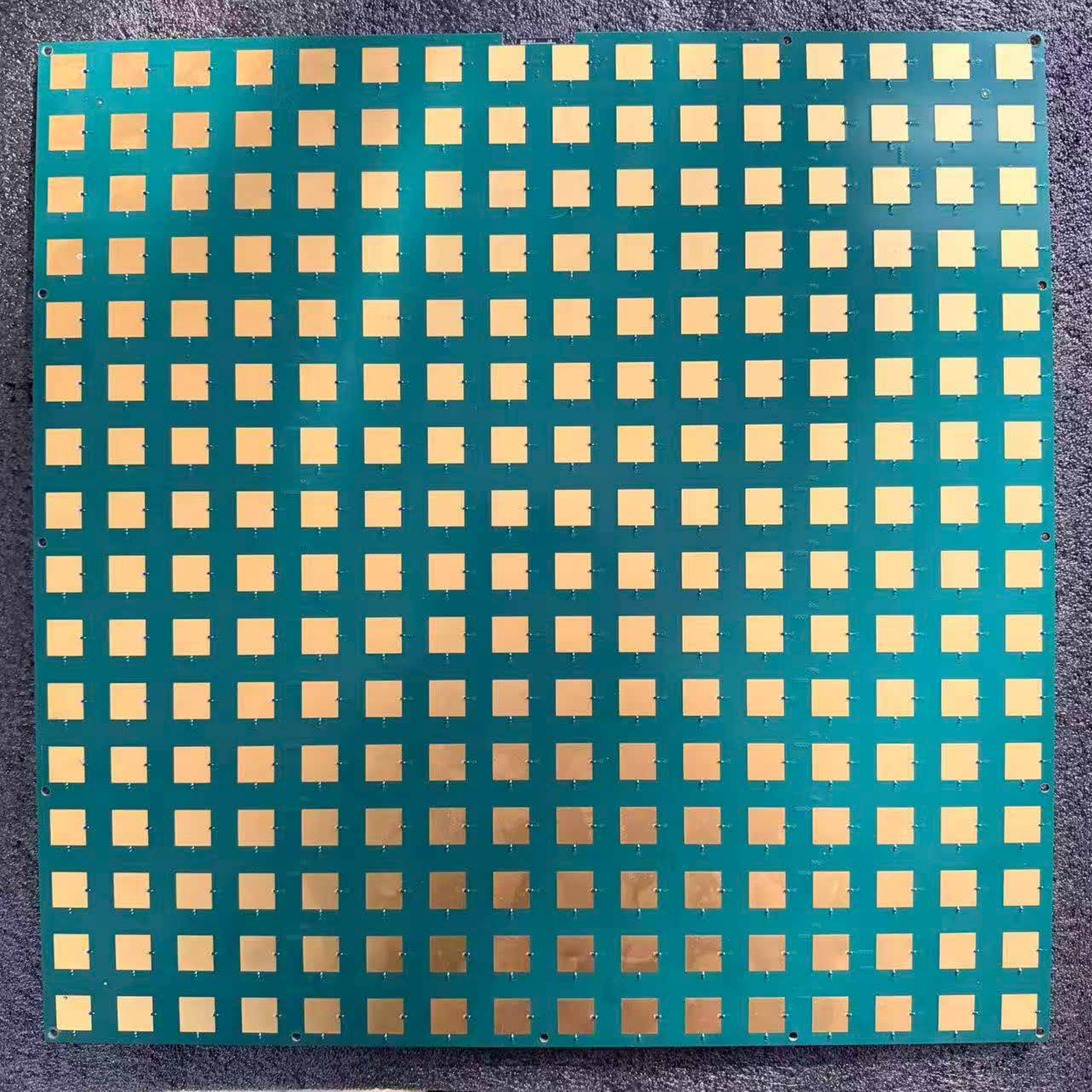}}\hfill
  \subfloat[Back view\label{fig:ris_b}]{%
    \includegraphics[width=0.48\columnwidth]{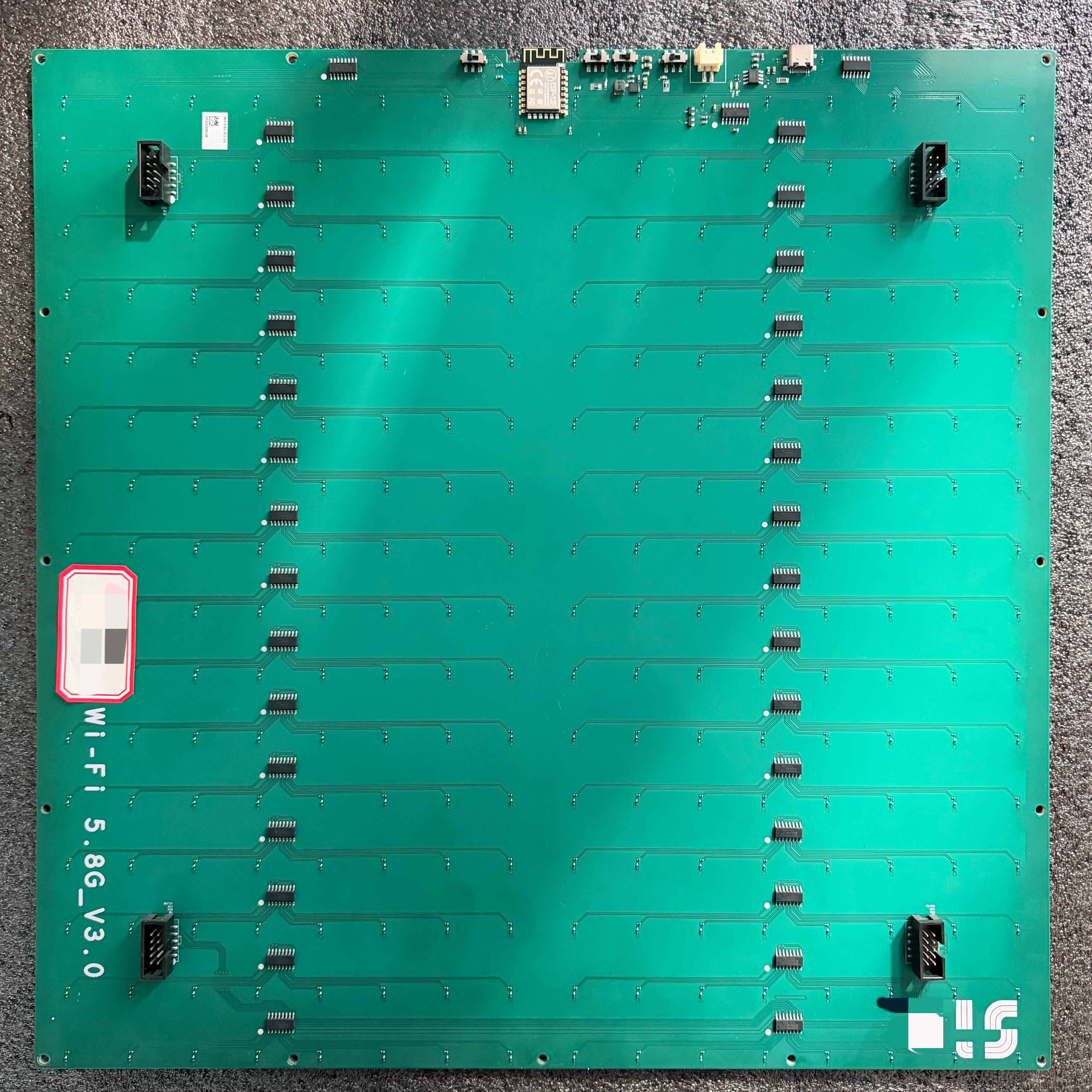}}
  \caption{Front and back views of the $16\times16$ RIS prototype operating at $5.8$~GHz with 1-bit phase control.}
  \label{fig:ris}
\end{figure}

Figures~\ref{fig:prototype_scenario} and \ref{fig:prototype_geometry} show the outdoor deployment geometry. A single target illuminates the RIS, and the receiver is connected to a universal software radio peripheral (USRP). The target--RIS distance is approximately $15.1~\mathrm{m}$, placing the incident field in the radiative far field and supporting the plane-wave model in \eqref{eq:measurement}. RIS--Rx propagation is represented by the geometry-dependent coefficients $\{\mathbf G_{m_x,m_y}\}$ in \eqref{eq:measurement}. For the $n$-th RIS configuration, the receiver records the corresponding complex baseband observation $y[n]$.

\begin{figure}
  \centering
  \includegraphics[width=0.85\linewidth]{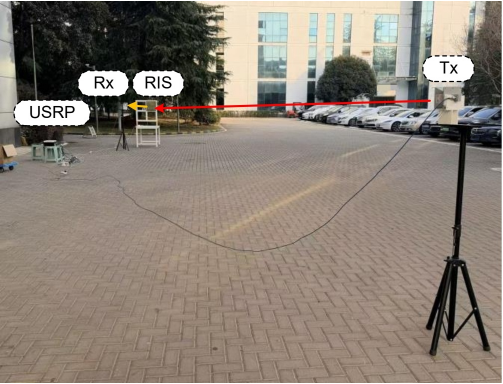}
  \caption{Outdoor prototype setup. A target illuminates the RIS, and a USRP-connected receiver captures the reflected signal.}
  \label{fig:prototype_scenario}
\end{figure}

\begin{figure}
  \centering
  \includegraphics[width=\linewidth]{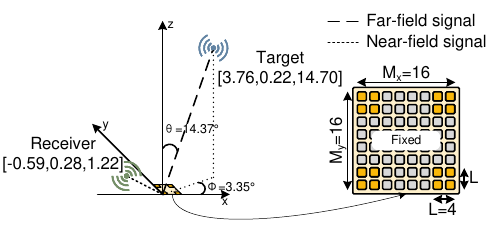}
  \caption{Geometry of the outdoor scenario. The target--RIS distance is approximately $15.1$~m, and the receiver location is represented through the prescribed RIS--Rx geometry.}
  \label{fig:prototype_geometry}
\end{figure}

We collect receiver observations under two RIS programming configurations:
\begin{itemize}
  \item \textit{Full-aperture programmable baseline:} All $16\times16$ elements participate in random 1-bit programming, corresponding to $\Omega_r^{(\mathrm{full})}=\Omega=\{1,\ldots,16\}\times\{1,\ldots,16\}$.
  \item \textit{Four-corner V-RIS configuration:} Only four $4\times 4$ corner subarrays are randomly programmed, corresponding to $\Omega_r^{(\mathrm{corner})}$ with $M_r=64$ programmable elements, or $25\%$ of the full aperture. The elements in $\Omega_v=\Omega\setminus\Omega_r^{(\mathrm{corner})}$ remain reflective with a fixed zero-phase state:
    \begin{equation}
      \begin{aligned}
      \Phi_{m_x,m_y}[n]&=
        \begin{cases}
          \text{random in }\{0,\pi\}, & \substack{(m_x,m_y)\in\\\Omega_r^{(\mathrm{corner})}},\\
          0, & (m_x,m_y)\in\Omega_v.
        \end{cases}
      \end{aligned}
      \label{eq:ris_phase_profile_poc}
    \end{equation}
\end{itemize}

The four-corner V-RIS configuration reduces the number of programmed elements from 256 to 64. Unlike the idealized sparse aperture in the numerical experiments, however, the non-corner elements remain physically present and contribute a reflected component under their fixed phase state. Because the geometry and fixed states are quasi-static, this contribution is approximately invariant across RIS configurations and can be treated as an additive bias.

\subsection{Reconstruction and Evaluation Procedure}
\paragraph{Practical Surface-Field Reconstruction}
The measured I/Q values contain amplitude and phase offsets, quasi-static coupling, residual leakage, and the fixed-state contribution described above. These effects are absorbed by the affine nuisance variables in \eqref{eq:BI_loss}, where the additive term $\nu$ represents contributions that are approximately constant across RIS configurations. This approximation is appropriate for the present quasi-static geometry.

We use the bias-invariant receiver-domain loss $\mathcal{L}_{\mathrm{data}}$ in \eqref{eq:BI_loss} when applying Algorithm~\ref{alg:inr_recon}. The closed-form alignment removes the configuration-invariant contribution and compensates for the unknown global complex gain at every optimization iteration, allowing the recurrence-constrained reconstruction to focus on the configuration-dependent spatial response.

The full-aperture baseline provides the reference field, whereas the four-corner V-RIS configuration reconstructs the $16\times16$ field from four programmable corner subarrays. Both use the same two-stage optimization. 

\paragraph{Evaluation Metrics}
We estimate the DoA from $\widehat{\mathbf H}$ using the Bartlett spectrum in \eqref{eq:bartlett_spectrum}. Let $(\theta,\phi)$ denote the ground-truth elevation and azimuth, and let $(\widehat{\theta}^{(q)},\widehat{\phi}^{(q)})$ denote their estimates for configuration $q$, where $q\in\{\mathrm{full},\mathrm{corner}\}$ identifies the full-aperture programmable baseline and the four-corner V-RIS configuration, respectively. For the single incident target, the coordinate-wise DoA error pair is
\begin{equation}
  \mathbf{e}^{(q)}
  =\big(e_{\theta}^{(q)},e_{\phi}^{(q)}\big)
  =\big(\widehat{\theta}^{(q)}-\theta,\
  \widehat{\phi}^{(q)}-\phi\big),
  \label{eq:doa_err_def}
\end{equation}
where $e_{\theta}^{(q)}$ and $e_{\phi}^{(q)}$ are the elevation and azimuth errors, respectively. To quantify the accuracy loss due to sparse programmability, we define the absolute error increment relative to the full-aperture programmable baseline as
\begin{equation}
  \Delta|\mathbf{e}| = \big(
    |e_{\theta}^{(\mathrm{corner})}|-|e_{\theta}^{(\mathrm{full})}|,
    |e_{\phi}^{(\mathrm{corner})}|-|e_{\phi}^{(\mathrm{full})}|
  \big).
  \label{eq:delta_abs_err}
\end{equation}
Thus, $\Delta|\mathbf{e}|$ denotes the additional absolute error introduced by four-corner programming relative to the full-aperture baseline.


\subsection{Prototype Results}
\paragraph{DoA Accuracy}

\begin{table}
  \centering
  \caption{Comparison of the outdoor prototype DoA estimates. Both configurations achieve sub-degree elevation and azimuth errors, while four-corner programming introduces absolute error increases of only $0.41^\circ$ and $0.21^\circ$ using $25\%$ of the programmable elements.}
  \label{tab:poc_doa}
\resizebox{\columnwidth}{!}{%
  \begin{tabular}{cccc}
    \hline
    & \textbf{Ground truth} & \textbf{Full-aperture baseline} & \textbf{Four-corner V-RIS} \\
    \hline
    DoA pair & $(14.37^\circ,3.35^\circ)$ & $(13.84^\circ,3.79^\circ)$ & $(13.43^\circ,2.70^\circ)$ \\
    $\mathbf{e}^{(q)}$ & -- & $(-0.53^\circ,0.44^\circ)$ & $(-0.94^\circ,-0.65^\circ)$ \\
    $\Delta|\mathbf{e}|$ & -- & -- & $(0.41^\circ,0.21^\circ)$ \\
    \hline
  \end{tabular}}
\end{table}

Table~\ref{tab:poc_doa} summarizes the DoA estimates and the corresponding metrics defined in \eqref{eq:doa_err_def} and \eqref{eq:delta_abs_err}. Both configurations achieve sub-degree errors in elevation and azimuth. Compared with full-aperture programming, the four-corner configuration increases the absolute errors by only $0.41^\circ$ and $0.21^\circ$, respectively, indicating that the reconstructed virtual-aperture surface field largely preserves the angular information despite using only $25\%$ programmable elements.

\paragraph{Spectrum Characteristics}
\begin{figure}
  \makebox[\linewidth][c]{%
    \subfloat[Full aperture\label{fig:doa_full}]{%
      \includegraphics[width=0.3\columnwidth]{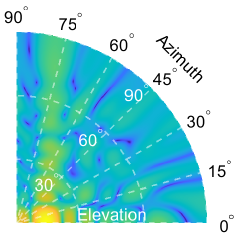}}
    \hspace{0.15\linewidth}
    \subfloat[Four corners\label{fig:doa_corner}]{%
      \includegraphics[width=0.3\columnwidth]{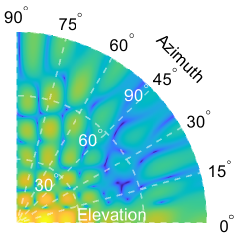}}
  }
  \caption{Comparison of the measured 2D Bartlett spectra obtained using the full-aperture programmable baseline and the four-corner V-RIS configuration. Four-corner V-RIS preserves the dominant DoA peak.}
  \label{fig:doa_poc}
\end{figure}

Figure~\ref{fig:doa_poc} provides the complementary spectral view. The four-corner V-RIS reconstruction preserves the dominant peak location but exhibits a higher sidelobe floor and more ripple-like artifacts than the full-aperture programmable baseline. These artifacts are consistent with the greater sensitivity of sparse programmability to phase-switching coupling, element-response nonuniformity, and configuration-to-configuration propagation fluctuations. The result therefore supports the DoA estimate while also showing the practical spectral cost of reducing the number of programmed elements.




\section{Conclusion}
This paper presented V-RIS, a virtual-aperture RIS framework for high-resolution 2D DoA estimation with sparse physical deployment. V-RIS addresses the hardware-cost and aperture-resolution tension by using four corner subarrays and a single-antenna receiver to reconstruct a virtual-aperture surface field from RIS-coded scalar receiver observations. The reconstruction is enabled by the finite-order spatial recurrence of far-field surface fields, which enforces propagation consistency when extending the surface field from deployed corner regions to virtual aperture positions. We incorporated this condition into a unified INR-based reconstruction formulation with a two-stage joint optimization procedure, and further introduced a bias-invariant receiver-domain loss to suppress quasi-static hardware distortions and configuration-invariant multipath contributions in practical receiver observations. Numerical experiments show that V-RIS approaches the DoA accuracy of a full-aperture benchmark, resolves the target directions more reliably than the compared reconstruction baselines, and benefits from the four-corner layout because it provides both local recurrence samples and long aperture baselines. The outdoor prototype further shows that V-RIS preserves sub-degree elevation and azimuth accuracy while programming only $25\%$ of the RIS elements, with only small error increases relative to full-aperture programming. These findings indicate that virtual-aperture surface-field reconstruction can decouple the RIS sensing aperture from the programmable-element count, offering a practical path toward lower-cost high-resolution RIS-aided DoA estimation. Future work will extend this principle to near-field, wideband, dynamic, and stronger-multipath scenarios while reducing the measurement and optimization overhead.

\appendices
\section{Proof of Affine Nuisance Elimination}
\label{Appendix:proof_LS}
\begin{proof}[Proof of Proposition~\ref{prop:affine_alignment}]
Fix $\Theta$ and abbreviate $\widehat{\mathbf y}=\widehat{\mathbf y}(\Theta;\Omega_r)$. Using the notation of Proposition~\ref{prop:affine_alignment}, define $\mathbf y_c=\mathbf P\mathbf y$ and $\widehat{\mathbf y}_c=\mathbf P\widehat{\mathbf y}$. Then
\[
\mathbf y-\rho\widehat{\mathbf y}-\nu\mathbf 1_N
=\mathbf y_c-\rho\widehat{\mathbf y}_c+(\mu_y-\rho\mu_{\widehat y}-\nu)\mathbf 1_N.
\]

Since the centered vectors are orthogonal to $\mathbf 1_N$,
\[
\left\|\mathbf y-\rho\widehat{\mathbf y}-\nu\mathbf 1_N\right\|_2^2
=\left\|\mathbf y_c-\rho\widehat{\mathbf y}_c\right\|_2^2+N\left|\mu_y-\rho\mu_{\widehat y}-\nu\right|^2.
\]

Hence, for fixed $\rho$, $\nu^\star(\rho)=\mu_y-\rho\mu_{\widehat y}$. The remaining problem is $\min_{\rho\in\mathbb C}\left\|\mathbf y_c-\rho\widehat{\mathbf y}_c\right\|_2^2$.
Since Proposition~\ref{prop:affine_alignment} assumes
$\mathbf P\widehat{\mathbf y}=\widehat{\mathbf y}_c\ne\mathbf 0$,
the denominator $\|\widehat{\mathbf y}_c\|_2^2$ is nonzero, and the unique
least-squares solution is
\[
\rho^\star(\Theta;\Omega_r)
=\frac{\widehat{\mathbf y}_c^{\mathrm H}\mathbf y_c}{\|\widehat{\mathbf y}_c\|_2^2}
=\frac{\widehat{\mathbf y}^{\mathrm H}\mathbf P\mathbf y}{\widehat{\mathbf y}^{\mathrm H}\mathbf P\widehat{\mathbf y}},
\]
\[
\nu^\star(\Theta;\Omega_r)
=\mu_y-\rho^\star(\Theta;\Omega_r)\mu_{\widehat y}.
\]
\end{proof}

\bibliographystyle{IEEEtran}
\bibliography{main}

\end{document}